\documentclass[11pt]{article}

\usepackage[margin=1in,footskip=0.25in]{geometry}
\usepackage{hyperref}
\usepackage{appendix}
\usepackage{graphicx}
\usepackage{xcolor}
\usepackage{enumerate}
\usepackage{subcaption}
\usepackage{multirow}

\usepackage[numbers,sort&compress]{natbib}

\usepackage{amsmath}
\usepackage{amsthm}
\usepackage{amssymb}
\usepackage{amsfonts}
\usepackage{bm}
\usepackage{bbm}
\usepackage{authblk}

\usepackage[ruled,vlined]{algorithm2e}
\SetKwInput{KwInput}{Input}
\SetKwInput{KwProcedure}{Procedure}
\SetKwInput{KwOutput}{Output}

\newcommand{\indep}{\raisebox{0.05em}{\rotatebox[origin=c]{90}{$\models$}}\,}
\newcommand{\reals}{\mathbb{R}}
\renewcommand{\P}[1]{\mathbb{P}\left[#1\right]}

\newcommand{\E}[1]{\mathbb{E}\left[#1\right]}
\newcommand{\I}[1]{\mathbbm{1}\left[#1\right]}

\newcommand{\alphalo}[0]{\alpha/2}
\newcommand{\alphaup}[0]{1-\alpha/2}
\newcommand{\alphalorm}[0]{\alpha/2}
\newcommand{\alphauprm}[0]{1-\alpha/2}

\newtheorem{definition}{Definition}
\newtheorem{assumption}{Assumption}
\newtheorem{theorem}{Theorem}

\title{A comparison of some conformal quantile regression methods}
\author[1]{Matteo Sesia}
\author[1,2]{Emmanuel J. Cand{\`e}s}
\affil[1]{Department of Statistics, Stanford University}
\affil[2]{Department of Mathematics, Stanford University}

\begin{document}
\maketitle

\begin{abstract}
  We compare two recently proposed methods that combine ideas from conformal inference and quantile regression to produce locally adaptive and marginally valid prediction intervals under sample exchangeability (Romano et al., 2019 \cite{romano2019conformalized}; Kivaranovic et al., 2019 \cite{kivaranovic2019adaptive}). First, we prove that these two approaches are asymptotically efficient in large samples, under some additional assumptions. Then we compare them empirically on simulated and real data. Our results demonstrate that the method in Romano et al.~(2019) typically yields tighter prediction intervals in finite samples. Finally, we discuss how to tune these procedures by fixing the relative proportions of observations used for training and conformalization.
 \end{abstract}

\section{Introduction} \label{sec:introduction}

\subsection{Background and motivation}

Given a set of $n$ points $\{(X_i,Y_i)\}_{i=1}^{n}$, with $Y_{i} \in \reals$ and $X_{i} \in \reals^d$, we consider the problem of constructing a prediction interval for a new point~$Y_{n+1}$ based on the observed value of~$X_{n+1}$, assuming only that $\{(X_i,Y_i)\}_{i=1}^{n+1}$ are drawn exchangeably from some common distribution~$P_{XY}$. There exist a vast selection of statistical and machine learning algorithms that can provide approximate answers to this question~\cite{papadopoulos2001confidence,wager2014confidence}. However, the uncertainty in any of their predictions cannot be quantified without making strong assumptions and large-sample asymptotic approximations that may not be easily justifiable in applications.
Conformal inference~\cite{vovk1999machine,vovk2005algorithmic,vovk2009line,papadopoulos2002inductive,papadopoulos2007,papadopoulos2008,papadopoulos2008inductive,papadopoulos2011,lei2018distribution} addresses this problem by constructing an exact \textit{marginal} prediction interval $\hat{C}_{\alpha}(X_{n+1})$ such that
\begin{align} \label{eq:marginal-coverage}
  \P{Y_{n+1} \in \hat{C}_{\alpha}(X_{n+1})} \geq 1-\alpha,
\end{align}
while relying only on the exchangeability of the $n+1$ points.
This interval is said to be \textit{marginal} because all variables in \eqref{eq:marginal-coverage} are treated as random, including $(X_{n+1}, Y_{n+1})$ and the data used to train~$\hat{C}$. Therefore, it is not guaranteed that the interval will cover~$Y_{n+1}$ conditional on a particular observed value of~$X_{n+1}$, or a fixed prediction model~$\hat{C}$. Despite this limitation, conformal prediction intervals are attractive because their coverage is guaranteed on average regardless of the distribution of the data.

Ideally, prediction intervals should be as narrow as possible while maintaining coverage. Let us denote by $q_{\alpha}(x_{n+1})$ the $\alpha$-th quantile of the conditional distribution of $Y$ given $X_{n+1}=x_{n+1}$. Then a desirable oracle prediction interval would be
\begin{align} \label{eq:oracle-interval}
  C^{\text{oracle}}_{\alpha}(X_{n+1}) = [q_{\alphalo}(X_{n+1}), q_{\alphaup}(X_{n+1})].
\end{align}
By construction, this is the narrowest \textit{symmetric} prediction interval that has valid coverage conditional on the value of $X_{n+1}$. Here, we say that a prediction interval is symmetric if $Y_{n+1}$ is equally likely to be smaller or larger than predicted. Unfortunately, the oracle interval in \eqref{eq:oracle-interval} is unachievable in practice because we do not know $P_{Y \mid X}$. The goal of conformal quantile regression~\cite{romano2019conformalized} is to form a practical prediction interval $\hat{C}_{\alpha}$ that estimates~\eqref{eq:oracle-interval} as closely as possible while satisfying~\eqref{eq:marginal-coverage} exactly.
In this work, we compare theoretically and empirically the method from~\cite{romano2019conformalized} with a similar approach that was proposed independently in~\cite{kivaranovic2019adaptive}.

\subsection{Conformal quantile regression} \label{sec:cqr}

Throughout this paper, we follow the split-conformal approach to conformal inference~\cite{papadopoulos2002inductive,papadopoulos2008inductive,lei2018distribution} adopted in~\cite{romano2019conformalized} and~\cite{kivaranovic2019adaptive}, since it is computationally feasible even with large data. The first step of the conformal quantile regression method in~\cite{romano2019conformalized} is to split the data samples into two disjoint subsets, $\mathcal{I}_1$ and $\mathcal{I}_2$. Lower and upper quantile regression functions, namely $\hat{q}_{\alphalo}, \hat{q}_{\alphaup} : \reals^d \to \reals$, are fitted on the observations in $\mathcal{I}_1$. Any algorithm can be employed for this purpose; for example, one may rely on linear regression~\cite{koenker1978regression}, neural networks~\cite{taylor2000quantile} or random forests~\cite{meinshausen2006quantile}. In any case, this algorithm is treated as a black box. The estimated quantile functions are used to compute a \textit{conformity score} for each $i \in \mathcal{I}_2$:
\begin{align} \label{eq:conformity-score-1}
  E_i^{\text{CQR}} = \max \left\{ \hat{q}_{\alphalo}(X_i)-Y_i, Y_i-\hat{q}_{\alphaup}(X_i)\right\}.
\end{align}
Then, with $\hat{Q}_{1-\alpha}(E^{\text{CQR}}; \mathcal{I}_2)$ defined as the $\lceil (1-\alpha) (|\mathcal{I}_2| +1)| \rceil$-th largest element of $\{E_i\}_{i \in \mathcal{I}_2}$, the conformal prediction interval for $X_{n+1}$ is given by
\begin{align} \label{eq:conformal-interval-1}
\begin{split}
  \hat{C}^{\text{CQR}}_\alpha(X_{n+1})
  & = \left[ \hat{q}_{\alphalo}(X_{n+1}) - \hat{Q}_{1-\alpha}(E^{\text{CQR}}; \mathcal{I}_2) , \hat{q}_{\alphaup}(X_{n+1}) + \hat{Q}_{1-\alpha}(E^{\text{CQR}}; \mathcal{I}_2) \right]
\end{split}
\end{align}
This method is summarized in Algorithm~\ref{alg:cqr}, where it is denoted as CQR. It is shown in~\cite{romano2019conformalized} that $\hat{C}^{\text{CQR}}_\alpha(X_{n+1})$ has marginal coverage at level $1-\alpha$.

\begin{algorithm}[!htb]
  \SetAlgoLined
  \KwInput{}
    data $\{(X_i,Y_i)\}_{i=1}^{n}$, covariates for new sample $X_{n+1}$; \\
    proportion of data for training $\gamma \in (0,1)$; \\
    quantile regression algorithm $\hat{q}$; \\
    conformalization method $\psi \in \{\text{CQR}, \text{CQR-m}, \text{CQR-r}\}$; \\
    coverage level $\alpha \in (0,1)$.

    \vspace{5pt}
    \KwProcedure{}
    randomly split $\{1,\ldots,n\}$ into $\mathcal{I}_1, \mathcal{I}_2$, of size $|\mathcal{I}_1|=\gamma n$, $|\mathcal{I}_2|=n-|\mathcal{I}_1|$;\\
    fit the quantile regression functions $\hat{q}_{\alphalo}$ and $\hat{q}_{\alphaup}$ on $\{(X_i,Y_i)\}_{i \in \mathcal{I}_1}$; \\
    \uIf{$\psi = $ \normalfont{CQR} } {
      compute the conformity scores $E^{\text{CQR}}_i$ for each $i \in \mathcal{I}_2$, as in~\eqref{eq:conformity-score-1};\\
      compute $\hat{Q}_{1-\alpha}(E^{\text{CQR}}; \mathcal{I}_2)$; \\
      compute the prediction interval $\hat{C}_\alpha(X_{n+1})$, as in~\eqref{eq:conformal-interval-1}.
    }
    \uElseIf{$\psi = $ \normalfont{CQR-m} } {
    fit the median regression function $\hat{q}_{1/2}$ on $\{(X_i,Y_i)\}_{i \in \mathcal{I}_1}$; \\
    compute the conformity scores $E_i^{\text{CQR-m}}$ for each $i \in \mathcal{I}_2$, as in~\eqref{eq:conformity-score-2};\\
    compute $\hat{Q}_{1-\alpha}(E^{\text{CQR-m}}; \mathcal{I}_2)$; \\
      compute the prediction interval $\hat{C}_\alpha(X_{n+1})$, as in~\eqref{eq:conformal-interval-2}.
    }
    \uElseIf{$\psi =$ \normalfont{CQR-r} } {
    compute the conformity scores $E_i^{\text{CQR-r}}$ for each $i \in \mathcal{I}_2$, as in~\eqref{eq:conformity-score-3};\\
    compute $\hat{Q}_{1-\alpha}(E^{\text{CQR-r}}; \mathcal{I}_2)$; \\
      compute the prediction interval $\hat{C}_\alpha(X_{n+1})$, as in~\eqref{eq:conformal-interval-3}.
    }

  \KwOutput{}
  A prediction interval $\hat{C}_\alpha(X_{n+1})$.

  \caption{Conformal quantile regression}
  \label{alg:cqr}
\end{algorithm}

The method described in~\cite{kivaranovic2019adaptive} differs from CQR in the choice of the conformity scores, as outlined in Algorithm~\ref{alg:cqr} as CQR-m. Instead of~\eqref{eq:conformity-score-1}, one computes\footnote{Note that we present CQR-m with a slightly different notation than in~\cite{kivaranovic2019adaptive} to facilitate the comparison.}
\begin{align} \label{eq:conformity-score-2}
  E_i^{\text{CQR-m}} = \max \left\{
  \frac{\hat{q}_{\alphalo}(X_i)-Y_i}{\hat{q}_{1/2}(X_i)-\hat{q}_{\alphalo}(X_i)},
  \frac{Y_i-\hat{q}_{\alphaup}(X_i)}{\hat{q}_{\alphaup}(X_i)-\hat{q}_{1/2}(X_i)}
  \right\},
\end{align}
where $\hat{q}_{1/2}$ indicates an estimated median regression function obtained with the same black-box algorithm as $\hat{q}_{\alphalo}$ and $\hat{q}_{\alphaup}$. Then the conformal prediction interval for $X_{n+1}$ is given by:
\begin{align} \label{eq:conformal-interval-2}
  \begin{split}
    \hat{C}^{\text{CQR-m}}_\alpha(X_{n+1})
    & = \left[ \hat{q}_{\alphalo}(X_{n+1}) - \hat{\Delta}_{\alpha,\text{lo}}^{\text{CQR-m}},
      \hat{q}_{\alphaup}(X_{n+1}) + \hat{\Delta}_{\alpha,\text{up}}^{\text{CQR-m}} \right] \\
    \hat{\Delta}_{\alpha, \text{lo}}^{\text{CQR-m}} & = \hat{Q}_{1-\alpha}(E^{\text{CQR-m}}; \mathcal{I}_2) \left[ \hat{q}_{1/2}(X_{n+1}) - \hat{q}_{\alphalo}(X_{n+1}) \right], \\
    \hat{\Delta}_{\alpha, \text{up}}^{\text{CQR-m}} & = \hat{Q}_{1-\alpha}(E^{\text{CQR-m}}; \mathcal{I}_2) \left[ \hat{q}_{\alphaup}(X_{n+1}) - \hat{q}_{1/2}(X_{n+1}) \right].
\end{split}
\end{align}
One can show that $\hat{C}^{\text{CQR-m}}_\alpha(X_{n+1})$ also has marginal coverage at level $1-\alpha$~\cite{kivaranovic2019adaptive}.

We also find it interesting to consider a modified version of CQR-m that does not require estimating the regression median.\footnote{This was first suggested by Yaniv Romano through personal communication.} This third approach, listed in Algorithm~\ref{alg:cqr} as CQR-r, is based on the following conformity scores:
\begin{align} \label{eq:conformity-score-3}
  E_i^{\text{CQR-r}} = \max \left\{
  \frac{\hat{q}_{\alphalo}(X_i)-Y_i}{\hat{q}_{\alphaup}(X_i)-\hat{q}_{\alphalo}(X_i)},
  \frac{Y_i-\hat{q}_{\alphaup}(X_i)}{\hat{q}_{\alphaup}(X_i)-\hat{q}_{\alphalo}(X_i)}
  \right\}.
\end{align}
The CQR-r prediction intervals are
\begin{align} \label{eq:conformal-interval-3}
\begin{split}
    \hat{C}^{\text{CQR-r}}_\alpha(X_{n+1})
    & = \left[ \hat{q}_{\alphalo}(X_{n+1}) - \hat{\Delta}_{\alpha}^{\text{CQR-r}},
      \hat{q}_{\alphaup}(X_{n+1}) + \hat{\Delta}_{\alpha}^{\text{CQR-r}} \right] \\
    \hat{\Delta}_{\alpha}^{\text{CQR-r}} & = \hat{Q}_{1-\alpha}(E^{\text{CQR-r}}; \mathcal{I}_2) \left[ \hat{q}_{\alphaup}(X_{n+1}) - \hat{q}_{\alphalo}(X_{n+1}) \right].
\end{split}
\end{align}

It is easy to show that $\hat{C}^{\text{CQR-r}}_\alpha(X_{n+1})$ also attains marginal coverage at level $1-\alpha$. A proof is omitted because it would be identical to those in~\cite{romano2019conformalized} and~\cite{kivaranovic2019adaptive}. CQR-r is similar in spirit to CQR-m, but it has a more direct interpretation: the conformity scores of CQR-r in~\eqref{eq:conformity-score-3} weight the distance of $Y$ from the corresponding prediction interval by the inverse width of the interval. Therefore, the conformalization expands or contracts the black-box prediction bands proportionally to their width, instead of adding a constant shift as in CQR. Since it is not clear how the regression median $q_{1/2}$ should generally be related to the upper and lower $\alpha$-quantiles of $P_{Y \mid X}$, we find this approach slightly more intuitive than CQR-m.



\section{Theoretical analysis} \label{sec:theory}

We show that the output of Algorithm~\ref{alg:cqr} converges to the oracle bands in~\eqref{eq:oracle-interval} as $n$ grows, if the black-box quantile regression estimates are consistent and a few additional assumptions hold. This can be established for any of the three alternative types of conformity scores discussed in this paper, which are therefore asymptotically equivalent in this sense. 

\begin{assumption}[i.i.d.]\label{assumption:iid}
The points $\{(X_i,Y_i)\}_{i=1}^{n+1}$ are drawn i.i.d.~from some distribution $P_{XY}$.
\end{assumption}

\begin{assumption}[regularity] \label{assumption:regularity}
The probability density of the conformity scores, either in~\eqref{eq:conformity-score-1},~\eqref{eq:conformity-score-2} or~\eqref{eq:conformity-score-3}, depending on the conformalization method in Algorithm~\ref{alg:cqr}, is bounded away from zero in an open neighborhood of zero.
\end{assumption}

\begin{assumption}[consistency] \label{assumption:consistency}
  For simplicity, denote by $n$ the size of the training data set $\mathcal{I}_1$ used to fit the quantile regression functions  $\hat{q}$. Let $X$ be a new observation independent of $\mathcal{I}_1$. Then the assumption is that, for 
$n$ large enough,
  \begin{align*}
    \P{\E{\left( \hat{q}_{\alphalorm}(X) - q_{\alphalorm}(X) \right)^2 \mid \hat{q}_{\alphalorm}, \hat{q}_{\alphauprm}} \leq \eta_{n} }
    & \geq 1-\rho_{n}, \\
    \P{\E{\left( \hat{q}_{\alphauprm}(X) - q_{\alphauprm}(X) \right)^2 \mid \hat{q}_{\alphalorm}, \hat{q}_{\alphauprm}} \leq \eta_{n} }
    & \geq 1-\rho_{n},
  \end{align*}
  for some sequences $\eta_{n}=o(1)$ and $\rho_{n}=o(1)$, as $n \to \infty$.
\end{assumption}

Assumption~\ref{assumption:consistency} is similar to that used in~\cite{lei2018distribution} for mean regression estimators, and it is weaker than requiring $\hat{q}_{\alphalorm}(X) \overset{L^2}{\to} q_{\alphalorm}(X)$ and $\hat{q}_{\alphauprm}(X) \overset{L^2}{\to} q_{\alphauprm}(X)$ as $n \to \infty$, by Markov's inequality. 

\begin{theorem} \label{thm:oracle-approximation}
Under Assumptions~\ref{assumption:iid}--\ref{assumption:consistency}, the conformal quantile regression bands $\hat{C}_{\alpha}$ obtained with Algorithm~\ref{alg:cqr} satisfy
\begin{align*}
  L\left( \hat{C}_{\alpha}(X_{n+1}) \,\triangle\, C_{\alpha}^{\mathrm{oracle}}(X_{n+1}) \right) = o_{\mathbb{P}}(1),
\end{align*}
as $|\mathcal{I}_1|, |\mathcal{I}_2| \to \infty$. Here, $L(A)$ indicates the Lebesgue measure of the set $A$, and $\triangle$ is the symmetric difference operator, i.e., $A \,\triangle\, B = (A \setminus B) \cup (B \setminus A)$.
\end{theorem}

The proof of Theorem~\ref{thm:oracle-approximation} can be found in Appendix~\ref{sec:proofs} and is inspired by that of Theorem~3.4 in~\cite{lei2018distribution}, although the oracle and the conformalization methods considered here are different. Theorem~\ref{thm:oracle-approximation} establishes a stronger form of statistical efficiency for conformal quantile regression compared to the result in~\cite{lei2018distribution}, which assumes $Y = \mu(X) + \epsilon$, for some regression function~$\mu$, and homoscedastic noise~$\epsilon$. In general, the conformal prediction intervals described in~\cite{lei2018distribution} will not converge to those of our oracle if the noise is heteroscedastic, regardless of the consistency of the black-box regression estimator~$\hat{\mu}$. By contrast, conformal quantile regression is efficient in the sense that, under Theorem~\ref{thm:oracle-approximation}, the prediction bands converge to those of the oracle, which are the narrowest possible bands achieving conditional coverage. Finally, the asymptotic consistency assumption may be verified theoretically for some specific algorithms under certain conditions, e.g.~random forests~\cite{meinshausen2006quantile}. In any case, our result provides some theoretical backing to conformal quantile regression even if the assumptions cannot be verified in practice.

As an immediate corollary of Theorem~\ref{thm:oracle-approximation}, note that it also follows that conformal quantile regression bands have asymptotic conditional coverage, which we define as in~\cite{lei2018distribution}.

\begin{definition}[Asymptotic conditional coverage] We say that a sequence $\hat{C}_n$ of random prediction bands has asymptotic conditional coverage at the level $1-\alpha$ if there exists a sequence of random sets $\Lambda_n \subseteq \mathbb{R}^d$ such that $\P{X \in \Lambda_n} = 1-o_{\mathbb{P}}(1)$ and
\begin{align*}
  \sup_{x \in \Lambda_n} \left| \P{Y \in \hat{C}_n(x) \mid X=x} - (1-\alpha) \right| = o_{\mathbb{P}}(1).
\end{align*}
\end{definition}

Despite being asymptotically efficient under Assumptions~\ref{assumption:iid}--\ref{assumption:consistency}, the three conformalization methods in Algorithm~\ref{alg:cqr} typically perform differently with finite data, as discussed next.

\section{Empirical comparison} \label{sec:empirical}

The data and code used in this section are on \url{https://github.com/msesia/cqr-comparison}.

\subsection{Black-box quantile regression}

In the following, we utilize two alternative black-box quantile regressors, implemented and trained as in~\cite{romano2019conformalized}. The first procedure is based on quantile regression forests~\cite{meinshausen2006quantile}. We fit 1000 trees and set the other tuning parameters equal to their default values. The second black box is a neural network~\cite{taylor2000quantile} with three fully connected layers and ReLU non-linearities. We have chosen this design, which is slightly different from that in~\cite{kivaranovic2019adaptive}, because it leads to conformal prediction intervals that are tighter than those reported in~\cite{kivaranovic2019adaptive}. If the estimated lower and upper quantiles overlap, which may sometimes occur, we swap them. The nominal level of the black boxes is tuned so that their empirical coverage, estimated by cross-validation, is approximately equal to $1-2\alpha$. We have observed that this heuristic generally leads to tighter conformal intervals compared to those obtained by directly requesting the black boxes to estimate $q_{\alpha/2}$ and $q_{1-\alpha/2}$; see Section~\ref{sec:exp-data} and Appendix~\ref{app:plots} for empirical evidence. Throughout this section, we set $\alpha=0.1$.

\subsection{Experiments with artificial data}

We begin by considering the same experiment based on artificial data as in~\cite{kivaranovic2019adaptive}. We simulate $X \sim \text{Unif}([0,1]^d)$, for $d=100$, and $Y \in \reals$ from:
\begin{align} \label{eq:artificial-data}
  & Y = f(\beta'X) + \epsilon \sqrt{1+(\beta'X)^2} ,
\end{align}
where $f(x) = 2\sin(\pi x) + \pi x$, $\beta'=(1,1,1,1,1,0,\ldots,0)$ and $\epsilon$ is independent standard Gaussian noise. Here, we have access to a natural benchmark: the oracle that knows $P_{Y \mid X}$ exactly. It follows from~\eqref{eq:artificial-data} that the expected width of the oracle prediction bands is:
\begin{align*}
  \E{q_{\alphaup}(X)-q_{\alphalo}(X)}
  & = 2 \, \E{\sqrt{1+(\beta'X)^2}} Q_{\alphaup}(\epsilon) \approx 8.91,
\end{align*}
where $Q_{\alpha}(\epsilon)$ is the $\alpha$-quantile of the standard Gaussian distribution.

The performances of CQR, CQR-m, and CQR-r are compared in Figure~\ref{fig:artificial} as a function of the number of data points $n$. The proportion of observations used to train the black box is $3/4$, as in~\cite{kivaranovic2019adaptive}. The coverage and average width of the prediction bands is evaluated on an independent test set of size $20,000$. The experiment is repeated for 100 independent realizations of the data and of the test set. The width and coverage of the conformal prediction bands approach those of the oracle as the sample size increases. This suggests that the estimated black-box quantiles may be approximately consistent. However, CQR typically produces narrower bands compared to the other methods.

\begin{figure}[!htb]
  \centering
  \includegraphics[]{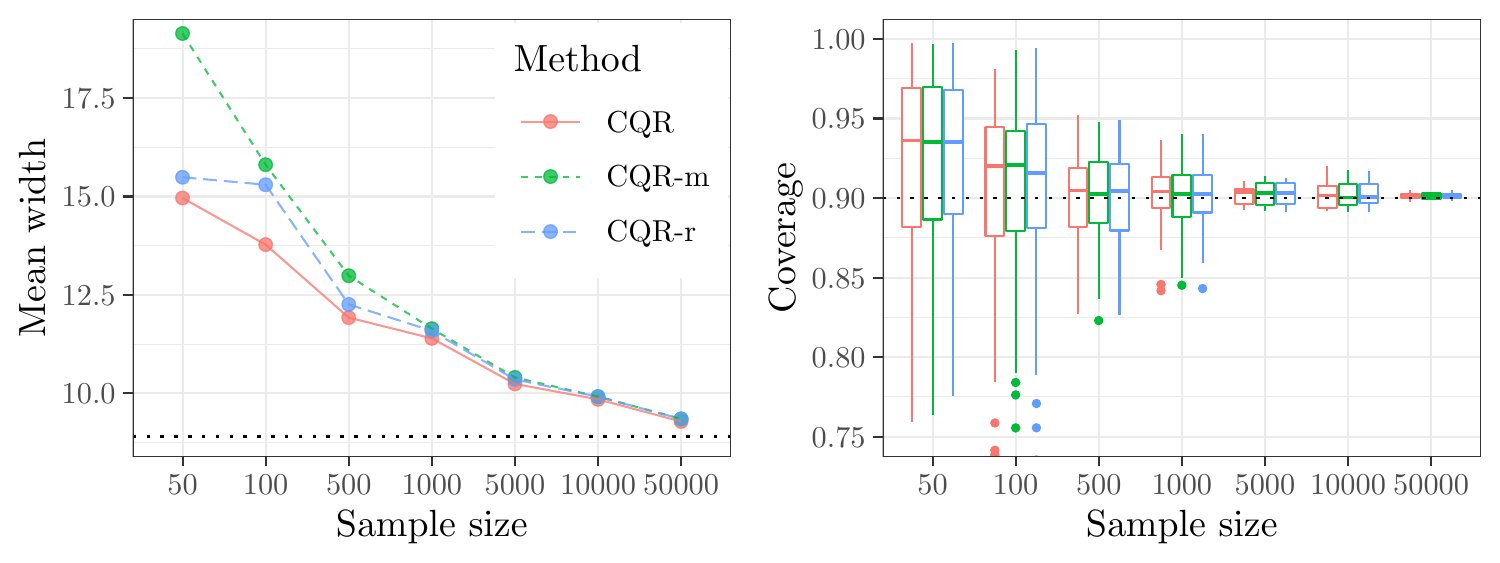}
  \caption{Conformal prediction bands obtained with different conformalization methods on artificial data, as a function of the sample size. The black dotted line on the left indicates the width of the oracle predictions. The black dotted line on the right indicates the nominal coverage level (90\%).}
  \label{fig:artificial}
\end{figure}

\subsection{Experiments with real data} \label{sec:exp-data}

We now apply Algorithm~\ref{alg:cqr} on the same data analyzed in~\cite{romano2019conformalized} and~\cite{kivaranovic2019adaptive}.\footnote{We have excluded the X-ray data in~\cite{kivaranovic2019adaptive} because we are unsure of how to replicate the pre-processing.}
 Some details about these data sets and information on the corresponding sources are summarized in Table~\ref{tab:datasets}. For all data sets except \textit{homes}, we randomly hold out 20\% of the samples for testing. Then we divide the remaining observations into two disjoint sets, $\mathcal{I}_1$ and $\mathcal{I}_2$, to train the black box and conformalize the prediction bands, respectively. The response variables $Y$ are standardized as in~\cite{romano2019conformalized} and~\cite{kivaranovic2019adaptive}.  We explore different values of the fraction of samples used for training: $|\mathcal{I}_1|=\gamma n$, with $\gamma \in \{0.1, 0.25, 0.5, 0.6, 0.7,0.8,0.9,0.95,0.98\}$. We are interested in this comparison because different values are used in~\cite{romano2019conformalized} and~\cite{kivaranovic2019adaptive}: $\gamma=0.5$ and $\gamma=0.75$, respectively. In the case of the \textit{homes} data, we follow in the footsteps of~\cite{kivaranovic2019adaptive}: first, we randomly hold out 3613 test samples; then, we train the black box on 15,000 samples and conformalize on the remaining~3000.
All experiments are repeated 10 times, starting from the data splitting.

\begin{table}[!htb]
\centering
\begin{tabular}{c|c|c|c|c}
\hline
Name & Description & $n$ & $d$ & Source \\
\hline
bike & bike sharing & 10886 & 18 &~\cite{data-bike}\\
bio & physicochemical properties of protein tertiary structures & 45730 & 9 &~\cite{data-bio}\\
blog & blog feedback & 52397 & 280 &~\cite{data-blog} \\
community & community and crime & 1994 & 100 &~\cite{data-community}\\
concrete & concrete compressive strength & 1030 & 8 &~\cite{data-concrete}\\
facebook 1 & facebook comment volume & 40948 & 53 &~\cite{data-facebook}\\
facebook 2 & facebook comment volume & 81311 & 53 &~\cite{data-facebook}\\
homes & sale prices of homes in King County, Washington & 21613 & 19 &~\cite{data-homes}\\
meps 19 & medical expenditure panel survey & 15785 & 139 &~\cite{data-meps19} \\
meps 20 & medical expenditure panel survey & 17541 & 139 &~\cite{data-meps20} \\
meps 21 & medical expenditure panel survey & 15656 & 139 &~\cite{data-meps21} \\
star & Tennessee's student-teacher achievement ratio & 2161 & 39 &~\cite{data-star}\\
\hline
\end{tabular}
\caption{Data sets for our empirical analysis, with numbers of samples ($n$) and features ($d$).}
\label{tab:datasets}
\end{table}

The test-set performances of CQR, CQR-m, and CQR-r are summarized in Tables~\ref{tab:result-summary-width} and~\ref{tab:result-summary-coverage}. These quantities correspond to the best choice of black box and the optimal value of the hyper-parameter~$\gamma$, defined separately for each algorithm. The CQR method consistently produces the narrowest valid prediction bands, while CQR-m and CQR-r are often comparable.

\begin{table}[!htb]
\centering
\begin{tabular}{|c|c|c|c|}
\hline
\multicolumn{1}{|c|}{ } & \multicolumn{3}{c|}{Width} \\
\cline{2-4}
\multirow{-2}{*}{Dataset} & CQR & CQR-r & CQR-m\\
\hline
bike & \textbf{0.503} (0.024) & 0.520 (0.024) & 0.521 (0.023)\\
\hline
bio & \textbf{0.995} (0.037) & 1.048 (0.049) & 1.114 (0.019)\\
\hline
blog & \textbf{1.269} (0.040) & 1.462 (0.148) & 1.351 (0.109)\\
\hline
community & \textbf{1.461} (0.116) & 1.548 (0.066) & 1.617 (0.063)\\
\hline
concrete & \textbf{0.378} (0.056) & 0.387 (0.063) & 0.384 (0.059)\\
\hline
facebook-1 & \textbf{1.117} (0.048) & 1.188 (0.043) & 1.164 (0.127)\\
\hline
facebook-2 & \textbf{1.110} (0.051) & 1.172 (0.066) & 1.116 (0.066)\\
\hline
homes & \textbf{0.477} (0.013) & 0.491 (0.013) & 0.492 (0.013)\\
\hline
meps-19 & \textbf{2.300} (0.164) & 2.349 (0.175) & 2.442 (0.364)\\
\hline
meps-20 & \textbf{2.309} (0.121) & 2.467 (0.313) & 2.467 (0.168)\\
\hline
meps-21 & \textbf{2.201} (0.076) & 2.273 (0.119) & 2.343 (0.337)\\
\hline
star & \textbf{0.179} (0.006) & 0.180 (0.010) & 0.181 (0.006)\\
\hline
\end{tabular}
\caption{Average width (and standard deviation) of the conformal quantile regression prediction bands obtained with different conformalization methods on the data sets listed in Table~\ref{tab:datasets}. The corresponding coverage is reported in Table~\ref{tab:result-summary-coverage}. The smallest value on each row is written in bold.}
\label{tab:result-summary-width}
\end{table}

\begin{table}[!htb]
\centering
\begin{tabular}{|c|c|c|c|}
\hline
\multicolumn{1}{|c|}{ } & \multicolumn{3}{c|}{Coverage} \\
\cline{2-4}
\multirow{-2}{*}{Dataset} & CQR & CQR-r & CQR-m\\
\hline
bike & 0.899 (0.012) & 0.901 (0.012) & 0.900 (0.012)\\
\hline
bio & 0.891 (0.012) & 0.893 (0.016) & 0.895 (0.008)\\
\hline
blog & 0.905 (0.003) & 0.901 (0.007) & 0.903 (0.004)\\
\hline
community & 0.896 (0.025) & 0.899 (0.025) & 0.902 (0.017)\\
\hline
concrete & 0.875 (0.061) & 0.879 (0.061) & 0.877 (0.059)\\
\hline
facebook-1 & 0.901 (0.006) & 0.898 (0.004) & 0.902 (0.002)\\
\hline
facebook-2 & 0.900 (0.003) & 0.900 (0.002) & 0.900 (0.002)\\
\hline
homes & 0.902 (0.009) & 0.904 (0.009) & 0.904 (0.009)\\
\hline
meps-19 & 0.902 (0.008) & 0.902 (0.007) & 0.900 (0.011)\\
\hline
meps-20 & 0.897 (0.004) & 0.898 (0.004) & 0.899 (0.006)\\
\hline
meps-21 & 0.899 (0.008) & 0.898 (0.009) & 0.897 (0.009)\\
\hline
star & 0.905 (0.024) & 0.904 (0.024) & 0.903 (0.020)\\
\hline
\end{tabular}
\caption{Average coverage (and standard deviation) of the prediction bands in Table~\ref{tab:result-summary-width}.}
\label{tab:result-summary-coverage}
\end{table}

The performances obtained with different black boxes and values of~$\gamma$ are reported in Figure~\ref{fig:community} for the \textit{community} data, and in Appendix~\ref{app:plots} for the other data sets. The results are shown as a function of~$\gamma$, which affects the average width of the prediction intervals as well as their variability. If $\gamma$ is small, the prediction intervals are not sufficiently adaptive because the black box cannot estimate the regression quantiles accurately. Larger values of~$\gamma$ may lead to tighter predictions on average, but at the cost of increased variability in the conditional coverage. In fact, the conditional coverage for new observations given the data may be lower than the expected marginal level, especially when $\gamma$ is very close to one and the sample size is not very large. The empirical results in Figure~\ref{fig:community} and Appendix~\ref{app:plots} suggest that fixing $\gamma \in [0.7,0.9]$ achieves a reasonable compromise for all data sets analyzed in this paper. This observation is also consistent with the choice in~\cite{kivaranovic2019adaptive}. 

The CQR-m method sometimes produces very wide intervals because the denominator in~\eqref{eq:conformity-score-2} can be close to zero (we added a small constant to prevent overflowing). An example is visible in the second plot in Figure~\ref{fig:community}, where some of the CQR-m prediction intervals based on a random forest black box are extremely large (hence the discontinuous vertical axis in Figure~\ref{fig:community}) when $\gamma \geq 0.9$. The CQR-r method is less susceptible to this problem because the denominator in the conformity scores in~\eqref{eq:conformity-score-3} is larger.

\begin{figure}[!htb]
  \centering
  \includegraphics[]{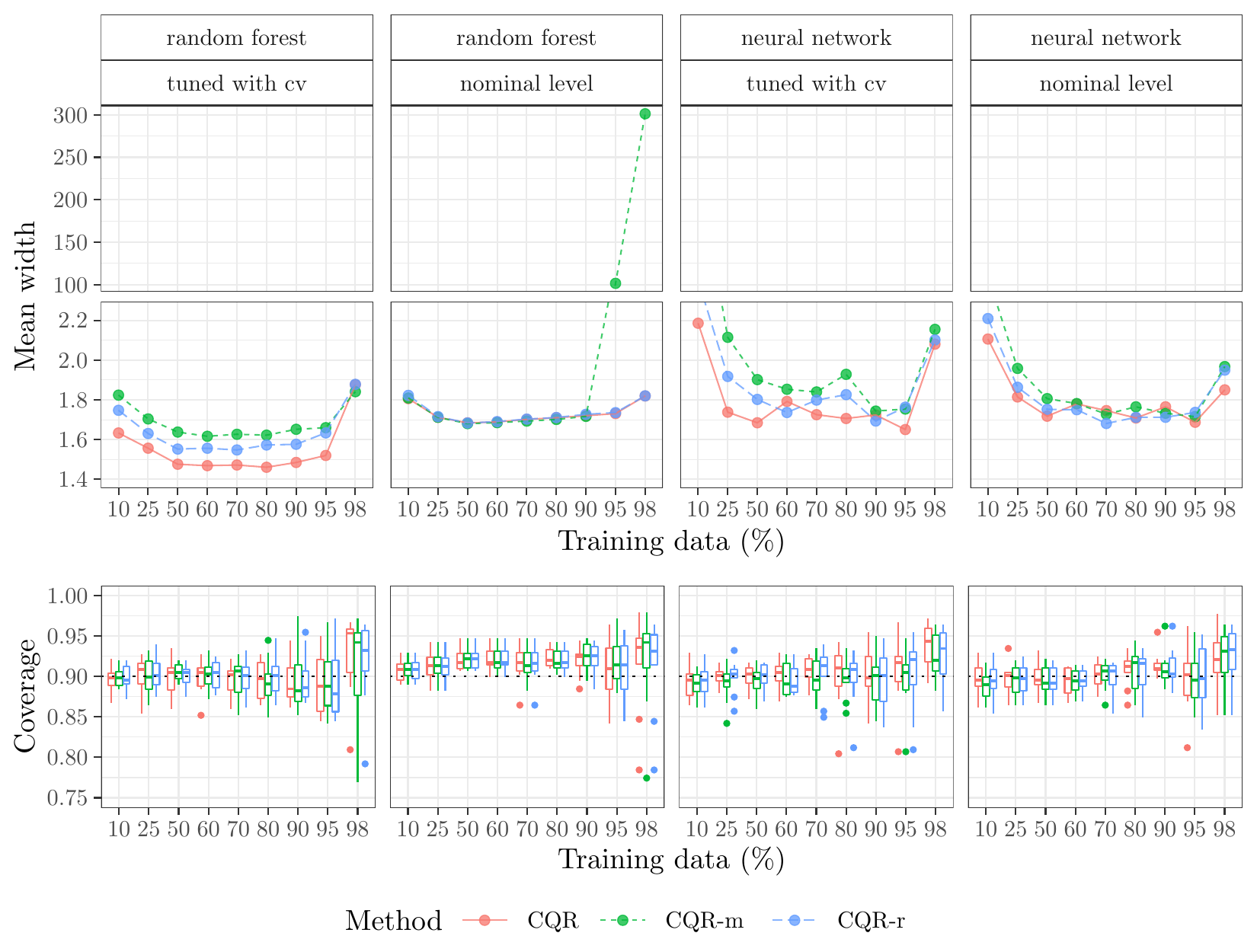}
  \caption{Conformal prediction bands obtained with different black boxes and conformalization methods on the \textit{community} data set from Table~\ref{tab:datasets}. The dotted black line in the lower plots indicates the nominal coverage level (90\%). A different black box is considered in each column. The vertical axis in the upper panels is discontinuous to facilitate the visualization of values on different scales.}
  \label{fig:community}
\end{figure}

\section{Conclusion} \label{sec:conclusion}

Early work on conformal prediction focused on estimating a mean regression function for $Y \mid X$ and building a fixed-width band around it~\cite{vovk1999machine,vovk2005algorithmic,vovk2009line,lei2018distribution}. Even though this strategy produces valid marginal prediction intervals regardless of $P_{Y \mid X}$, it is clearly designed with a homoscedastic regression model in mind and it may lead to be unnecessarily wide intervals in other cases. Locally-adaptive conformal prediction~\cite{papadopoulos2007,papadopoulos2008,papadopoulos2011,lei2018distribution} goes a step beyond this model by weighting the residuals according to a local estimate of their variance. Conformal quantile regression~\cite{romano2019conformalized} goes further by observing that the estimation of the regression mean is unnecessary if the ultimate goal is to build prediction intervals. This approach has already been shown to outperform earlier methods in practice~\cite{romano2019conformalized}.

In this paper, we have strengthened the case for conformal quantile regression by proving that it is asymptotically efficient in large samples, if the quantile regression estimates are consistent. The empirical comparison of three alternative conformity scores has shown that those proposed in~\cite{romano2019conformalized} are preferable because they typically lead to shorter prediction intervals in practice. Even though we have only explicitly considered symmetric intervals for simplicity, it is straightforward to generalize these methods to asymmetric intervals and conformity scores~\cite{romano2019conformalized}.
Finally, we have highlighted a bias-variance tradeoff in the choice of the proportion of data points used to train the black-box quantile regressors. Our empirical results show that it is usually better to invest more of the available data (between 70\% and 90\%, indicatively) to train the black-box than to conformalize the predictions. We hope that these results will be helpful to practitioners and may inspire others to develop even more powerful variations of conformal quantile regression.

\subsection*{Acknowledgements}

M.~S.~and E.~C.~are supported by the National Science Foundation under grant DMS 1712800.
E.~C.~is also supported by the Army Research Office under grant W911NF-17-1-0304.
We thank Yaniv Romano for helpful discussions, during which he suggested the CQR-r method.

\bibliography{bibliography}

\begin{thebibliography}{10}

\bibitem{romano2019conformalized}
Y.~Romano, E.~Patterson, and E.~J. Cand{\`e}s, ``Conformalized quantile
  regression,'' {\em arXiv preprint arXiv:1905.03222}, 2019.

\bibitem{kivaranovic2019adaptive}
D.~Kivaranovic, K.~D. Johnson, and H.~Leeb, ``Adaptive, distribution-free
  prediction intervals for deep neural networks,'' {\em arXiv preprint
  arXiv:1905.10634}, 2019.

\bibitem{papadopoulos2001confidence}
G.~Papadopoulos, P.~J. Edwards, and A.~F. Murray, ``Confidence estimation
  methods for neural networks: A practical comparison,'' {\em IEEE Transactions
  on Neural Networks}, vol.~12, no.~6, pp.~1278--1287, 2001.

\bibitem{wager2014confidence}
S.~Wager, T.~Hastie, and B.~Efron, ``Confidence intervals for random forests:
  The jackknife and the infinitesimal jackknife,'' {\em The Journal of Machine
  Learning Research}, vol.~15, no.~1, pp.~1625--1651, 2014.

\bibitem{vovk1999machine}
V.~Vovk, A.~Gammerman, and C.~Saunders, ``Machine-learning applications of
  algorithmic randomness,'' in {\em Proceedings of the Sixteenth International
  Conference on Machine Learning}, ICML '99, (San Francisco, CA, USA),
  pp.~444--453, Morgan Kaufmann Publishers Inc., 1999.

\bibitem{vovk2005algorithmic}
V.~Vovk, A.~Gammerman, and G.~Shafer, {\em Algorithmic Learning in a Random
  World}.
\newblock Berlin, Heidelberg: Springer-Verlag, 2005.

\bibitem{vovk2009line}
V.~Vovk, I.~Nouretdinov, A.~Gammerman, {\em et~al.}, ``On-line predictive
  linear regression,'' {\em The Annals of Statistics}, vol.~37, no.~3,
  pp.~1566--1590, 2009.

\bibitem{papadopoulos2002inductive}
H.~Papadopoulos, K.~Proedrou, V.~Vovk, and A.~Gammerman, ``Inductive confidence
  machines for regression,'' in {\em European Conference on Machine Learning},
  pp.~345--356, Springer, 2002.

\bibitem{papadopoulos2007}
H.~{Papadopoulos}, V.~{Vovk}, and A.~{Gammermam}, ``Conformal prediction with
  neural networks,'' in {\em 19th IEEE International Conference on Tools with
  Artificial Intelligence(ICTAI 2007)}, vol.~2, pp.~388--395, Oct 2007.

\bibitem{papadopoulos2008}
H.~Papadopoulos, A.~Gammerman, and V.~Vovk, ``Normalized nonconformity measures
  for regression conformal prediction,'' in {\em Proceedings of the 26th IASTED
  International Conference on Artificial Intelligence and Applications}, AIA
  '08, (Anaheim, CA, USA), pp.~64--69, ACTA Press, 2008.

\bibitem{papadopoulos2008inductive}
H.~Papadopoulos, ``Inductive conformal prediction: Theory and application to
  neural networks,'' in {\em Tools in Artificial Intelligence} (P.~Fritzsche,
  ed.), ch.~18, Rijeka: IntechOpen, 2008.

\bibitem{papadopoulos2011}
H.~Papadopoulos, V.~Vovk, and A.~Gammerman, ``Regression conformal prediction
  with nearest neighbours,'' {\em Journal of Artificial Intelligence Research},
  vol.~40, pp.~815--840, Jan. 2011.

\bibitem{lei2018distribution}
J.~Lei, M.~G’Sell, A.~Rinaldo, R.~J. Tibshirani, and L.~Wasserman,
  ``Distribution-free predictive inference for regression,'' {\em Journal of
  the American Statistical Association}, vol.~113, no.~523, pp.~1094--1111,
  2018.

\bibitem{koenker1978regression}
R.~Koenker and G.~Bassett~Jr, ``Regression quantiles,'' {\em Econometrica:
  Journal of the Econometric Society}, pp.~33--50, 1978.

\bibitem{taylor2000quantile}
J.~W. Taylor, ``A quantile regression neural network approach to estimating the
  conditional density of multiperiod returns,'' {\em Journal of Forecasting},
  vol.~19, no.~4, pp.~299--311, 2000.

\bibitem{meinshausen2006quantile}
N.~Meinshausen, ``Quantile regression forests,'' {\em Journal of Machine
  Learning Research}, vol.~7, no.~Jun, pp.~983--999, 2006.

\bibitem{data-bike}
``Bike sharing dataset data set.''
  \url{https://archive.ics.uci.edu/ml/datasets/bike+sharing+dataset}.
\newblock Accessed: July, 2019.

\bibitem{data-bio}
``Physicochemical properties of protein tertiary structure data set.''
  \url{https://archive.ics.uci.edu/ml/datasets/Physicochemical+Properties+of+Protein+Tertiary+Structure}.
\newblock Accessed: July, 2019.

\bibitem{data-blog}
``Blog{F}eedback data set.''
  \url{https://archive.ics.uci.edu/ml/datasets/BlogFeedback}.
\newblock Accessed: July, 2019.

\bibitem{data-community}
``Communities and crime data set.''
  \url{http://archive.ics.uci.edu/ml/datasets/communities+and+crime}.
\newblock Accessed: July, 2019.

\bibitem{data-concrete}
``Concrete compressive strength data set.''
  \url{http://archive.ics.uci.edu/ml/datasets/concrete+compressive+strength}.
\newblock Accessed: July, 2019.

\bibitem{data-facebook}
``Facebook comment volume data set..''
  \url{https://archive.ics.uci.edu/ml/datasets/Facebook+Comment+Volume+Dataset}.
\newblock Accessed: July, 2019.

\bibitem{data-homes}
``House sales in {K}ing {C}ounty, {USA}.''
  \url{https://www.kaggle.com/harlfoxem/housesalesprediction/metadata}.
\newblock Accessed: August, 2019.

\bibitem{data-meps19}
``Medical expenditure panel survey, panel 19.''
  \url{https://meps.ahrq.gov/mepsweb/data_stats/download_data_files_detail.jsp?cboPufNumber=HC-181}.
\newblock Accessed: July, 2019.

\bibitem{data-meps20}
``Medical expenditure panel survey, panel 20.''
  \url{https://meps.ahrq.gov/mepsweb/data_stats/download_data_files_detail.jsp?cboPufNumber=HC-181}.
\newblock Accessed: July, 2019.

\bibitem{data-meps21}
``Medical expenditure panel survey, panel 21.''
  \url{https://meps.ahrq.gov/mepsweb/data_stats/download_data_files_detail.jsp?cboPufNumber=HC-192
  }.
\newblock Accessed: July, 2019.

\bibitem{data-star}
C.~Achilles, H.~P. Bain, F.~Bellott, J.~Boyd-Zaharias, J.~Finn, J.~Folger,
  J.~Johnston, and E.~Word, ``Tennessee's student teacher achievement ratio
  ({STAR}) project,'' 2008.
\newblock Accessed: July, 2019.

\bibitem{vaart_1998}
A.~W. v.~d. Vaart, {\em Asymptotic Statistics}.
\newblock Cambridge Series in Statistical and Probabilistic Mathematics,
  Cambridge University Press, 1998.

\end{thebibliography}
\bibliographystyle{ieeetr}

\clearpage
\appendix
\section{Proofs} \label{sec:proofs}

\begin{proof}[Proof of Theorem~\ref{thm:oracle-approximation}]
  We begin by considering the case of CQR. For ease of notation and without loss of generality, assume that we have $2n$ data points and $n_1=n_2=n$. Then, we can equivalently rewrite Assumption~\ref{assumption:consistency} as follows:
  \begin{align*}
    \P{\E{\left( \hat{q}_{\alphalorm}(X) - q_{\alphalorm}(X) \right)^2 \mid \hat{q}_{\alphalorm}, \hat{q}_{\alphauprm}} \leq \frac{\eta_{n}}{2} }
    & \geq 1-\frac{\rho_{n}}{2}, \\
    \P{\E{\left( \hat{q}_{\alphauprm}(X) - q_{\alphauprm}(X) \right)^2 \mid \hat{q}_{\alphalorm}, \hat{q}_{\alphauprm}} \leq \frac{\eta_{n}}{2} }
    & \geq 1-\frac{\rho_{n}}{2},
  \end{align*}
  for $n$ large enough, $X \indep \mathcal{I}_1$ and 
  for some sequences $\eta_{n}=o(1)$ and $\rho_{n}=o(1)$, as $n \to \infty$.

Recall that the conformal quantile regression prediction band is defined as:
\begin{align*}
  \hat{C}^{\text{CQR}}_{\alpha}(X_{2n+1}) = \left[ \hat{q}_{\alphalo}(X_{2n+1}) - \hat{Q}_{1-\alpha}(E^{\text{CQR}};\mathcal{I}_2), \hat{q}_{\alphaup}(X_{2n+1}) + \hat{Q}_{1-\alpha}(E^{\text{CQR}};\mathcal{I}_2)\right],
\end{align*}
while the oracle band is:
\begin{align*}
  C^{\text{oracle}}_{\alpha}(X_{2n+1}) = \left[ q_{\alphalo}(X_{2n+1}), q_{\alphaup}(X_{2n+1})\right].
\end{align*}
It suffices to show:
\begin{align*}
  |\hat{q}_{\alphalo}(X_{2n+1}) - \hat{Q}_{1-\alpha}(E^{\text{CQR}};\mathcal{I}_2) - q_{\alphalo}(X_{2n+1})| & = o_{\mathbb{P}}(1), \\
  |\hat{q}_{\alphaup}(X_{2n+1}) + \hat{Q}_{1-\alpha}(E^{\text{CQR}};\mathcal{I}_2) - q_{\alphaup}(X_{2n+1})| & = o_{\mathbb{P}}(1).
\end{align*}
We will proceed in two steps, proving:
\begin{enumerate}[(i)]
  \item $|\hat{q}_{\alphalo}(X) - q_{\alphalo}(X)| = o_{\mathbb{P}}(1)$ and $|\hat{q}_{\alphaup}(X) - q_{\alphaup}(X)| = o_{\mathbb{P}}(1)$, for $X \indep \hat{q}_{\alphalo}, \hat{q}_{\alphaup}$;
  \item $|\hat{Q}_{1-\alpha}(E^{\text{CQR}};\mathcal{I}_2)| = o_{\mathbb{P}}(1)$.
\end{enumerate}
Then the proof will be completed by the triangle inequality.
\begin{enumerate}[(i)]
\item Define the random sets
\begin{align*}
  & B_{n,\text{up}} = \{x : |\hat{q}_{\alphaup}(x)-q_{\alphaup}(x)| \geq \eta_n^{1/3}\},
  & B_{n,\text{lo}} = \{x : |\hat{q}_{\alphalo}(x)-q_{\alphalo}(x)| \geq \eta_n^{1/3}\},
\end{align*}
 and $B_{n} = B_{n,\text{up}} \cup B_{n,\text{lo}}$. We can prove that for a new $X \indep \hat{q}_{\alphalo}, \hat{q}_{\alphaup}$,
\begin{align} \label{eq:p(Bn)}
\P{X \in B_n \mid \hat{q}_{\alphalo}, \hat{q}_{\alphaup}} \leq \eta_{n}^{1/3} + \rho_n.
\end{align}
In fact, in the event
\begin{align} \label{eq:event-consistent}
\left\{\E{\left( \hat{q}_{\alphalo}(X) - \hat{q}_{\alphalo}(X) \right)^2} \leq \frac{\eta_n}{2} \right\}
  \cap
  \left\{ \E{\left( \hat{q}_{\alphaup}(X) - \hat{q}_{\alphaup}(X) \right)^2} \leq \frac{\eta_n}{2} \right\},
\end{align}
we have:
\begin{align*}
  & \P{X \in B_n \mid \hat{q}_{\alphalo}, \hat{q}_{\alphaup}} \\
  & \qquad  = \P{X \in B_{n,\text{lo}} \cup B_{n,\text{up}} \mid \hat{q}_{\alphalo}, \hat{q}_{\alphaup}} \\
  & \qquad \leq \P{X \in B_{n,\text{lo}}  \mid \hat{q}_{\alphalo} } +
    \P{X \in B_{n,\text{up}} \mid \hat{q}_{\alphaup}} \\
  & \qquad = \P{ |\hat{q}_{\alphalo}(X)-q_{\alphalo}(X)| \geq \eta_n^{1/3} \mid \hat{q}_{\alphalo} } + \P{ |\hat{q}_{\alphaup}(X)-q_{\alphaup}(X)| \geq \eta_n^{1/3} \mid \hat{q}_{\alphaup} } \\
  & \qquad = \P{ |\hat{q}_{\alphalo}(X)-q_{\alphalo}(X)|^2 \geq \eta_n^{2/3} \mid \hat{q}_{\alphalo} } + \P{ |\hat{q}_{\alphaup}(X)-q_{\alphaup}(X)|^2 \geq \eta_n^{2/3} \mid \hat{q}_{\alphaup} } \\
  & \qquad \leq \eta_n^{-2/3} \E{\left( \hat{q}_{\alphalo}(X) - \hat{q}_{\alphalo}(X) \right)^2} + \eta_n^{-2/3} \E{\left( \hat{q}_{\alphaup}(X) - \hat{q}_{\alphaup}(X) \right)^2} \\
  & \qquad \leq \eta_n^{1/3}.
\end{align*}
Equation~\eqref{eq:p(Bn)} follows because the event in~\eqref{eq:event-consistent} has probability at least $1-\rho_n$, by Assumption~\ref{assumption:consistency}. This implies:
\begin{align*}
  & |\hat{q}_{\alphaup}(X) - q_{\alphaup}(X)| = o_{\mathbb{P}}(1),
  & |\hat{q}_{\alphalo}(X) - q_{\alphalo}(X)| = o_{\mathbb{P}}(1).
\end{align*}

\item With $B_n$ defined as above, consider the following further partition of the data in $\mathcal{I}_2$:
\begin{align*}
  & \mathcal{I}_{2,a} = \{ i \in \{n+1,\ldots,2n\} : X_i \in B_n^c\},
  & \mathcal{I}_{2,b} = \{ i \in \{n+1,\ldots,2n\} : X_i \in B_n\}.
\end{align*}
By definition, $ \mathcal{I}_{2} = \{n+1, \ldots, 2n\} = \mathcal{I}_{2,a} \cup \mathcal{I}_{2,b}$. Since $B_n$ only depends on the data in $\mathcal{I}_1$, it is independent of $(X_i,Y_i)$ for all $i \in \mathcal{I}_2$. Therefore, the size of $\mathcal{I}_{2,b}$ conditional on the data in $\mathcal{I}_{1}$ can be bounded using Hoeffding's inequality, thanks to Assumption~\ref{assumption:iid}. We already know that the probability that any $i \in \{n+1,\ldots,2n\}$ belongs to $\mathcal{I}_{2,b}$ is smaller than $\eta_n^{1/3}$. In particular, conditional on $\hat{q}_{\alphalo}$ and $\hat{q}_{\alphaup}$,
\begin{align*}
  \P{|\mathcal{I}_{2,b}| \geq n \eta_n^{1/3} + t }
  & \leq \P{\frac{1}{n} \sum_{i=n+1}^{2n} \I{X_i \in B_n}  \geq \P{X_i \in B_n} + \frac{t}{n} }
   \leq \exp \left( -\frac{2t^2}{n} \right).
\end{align*}
Choosing $t = c\sqrt{n \log n}$ leads to $|\mathcal{I}_{2,b}| = o_{\mathbb{P}}(n)$ because
\begin{align*}
  \P{|\mathcal{I}_{2,b}| \geq n \eta_n^{1/3} + c\sqrt{n \log n} }
  & \leq n^{-2c^2}.
\end{align*}
Now, define $\tilde{E}^{\text{CQR}}_i = \max\{ q_{\alphalo}(X_{n+1}) - Y, Y - q_{\alphaup}(X_{n+1}) \}$ for any $i \in \mathcal{I}_2$. By definition of $E^{\text{CQR}}_i$, for all $i \in \mathcal{I}_{2,a}$,
\begin{align*}
  E^{\text{CQR}}_i
  & = \max\{ \hat{q}_{\alphalo}(X_{n+1}) - Y, Y - \hat{q}_{\alphaup}(X_{n+1}) \} \\
  & = \max\{ \hat{q}_{\alphalo}(X_{n+1}) - q_{\alphalo}(X_{n+1}) + q_{\alphalo}(X_{n+1}) - Y, \\
  & \textcolor{white}{= \max\{\, } Y - q_{\alphaup}(X_{n+1}) + q_{\alphaup}(X_{n+1}) - \hat{q}_{\alphaup}(X_{n+1}) \} \\
  & \leq \max\{ \eta_n^{1/3} + q_{\alphalo}(X_{n+1}) - Y, Y - q_{\alphaup}(X_{n+1}) + \eta_n^{1/3} \} \\
  & = \eta_n^{1/3} + \max\{ q_{\alphalo}(X_{n+1}) - Y, Y - q_{\alphaup}(X_{n+1}) \} \\
  & = \eta_n^{1/3} + \tilde{E}^{\text{CQR}}_i.
\end{align*}
Proceeding similarly, one can also show that $E^{\text{CQR}}_i \geq \tilde{E}^{\text{CQR}}_i - \eta_n^{1/3}$. Hence, for all $i \in \mathcal{I}_{2,a}$,
\begin{align*}
  \left| E^{\text{CQR}}_i - \tilde{E}^{\text{CQR}}_i \right| \leq \eta_n^{1/3}.
\end{align*}
Therefore, all empirical quantiles of $E^{\text{CQR}}_i$ and $\tilde{E}^{\text{CQR}}_i$, for $i \in \mathcal{I}_{2,a}$, are within $\eta_n^{1/3}$.

Let $F_n$ and $\tilde{F}_n$ denote the empirical distributions of $E^{\text{CQR}}_i$ and $\tilde{E}^{\text{CQR}}_i$ for $i \in \mathcal{I}_2$, respectively. Define also $F_{n,a}$ and $\tilde{F}_{n,a}$ as the corresponding empirical distributions when $i$ is restricted to $\mathcal{I}_{2,a}$. For $n$ large enough, one can assume without loss of generality that $|\mathcal{I}_{2,b}|/|\mathcal{I}_{2,a}| \leq \alpha$ because $|\mathcal{I}_{2,b}| = o_{\mathbb{P}}(n)$. Then we can show that
\begin{align} \label{eq:quantile-sandwich}
  F^{-1}_{n,a}\left( 1- \frac{n\alpha}{|\mathcal{I}_{2,a}|} \right)
  & \leq F^{-1}_{n}\left( 1- \alpha \right)
    \leq F^{-1}_{n,a}\left( 1- \frac{n\alpha -|\mathcal{I}_{2,b}| }{|\mathcal{I}_{2,a}|} \right).
\end{align}
To prove the second inequality in \eqref{eq:quantile-sandwich}, note that if all the $E^{\text{CQR}}_i$, for $i \in \mathcal{I}_{2,b}$, are in the upper $\alpha$-quantile of $F_{n}$, then
\begin{align*}
  F^{-1}_{n,a}\left( 1- \frac{n\alpha -|\mathcal{I}_{2,b}| }{n} \frac{n}{|\mathcal{I}_{2,a}|} \right)
  & = F^{-1}_{n}\left( 1- \alpha \right).
\end{align*}
However, in general the quantiles of $F_{n,a}$ will be larger and
\begin{align*}
  F^{-1}_{n,a}\left( 1- \frac{n\alpha -|\mathcal{I}_{2,b}| }{|\mathcal{I}_{2,a}|} \right)
  & \geq F^{-1}_{n}\left( 1- \alpha \right).
\end{align*}
To prove the first inequality in \eqref{eq:quantile-sandwich}, note that if all the $E^{\text{CQR}}_i$ for $i \in \mathcal{I}_{2,b}$ are in the lower $1-\alpha$ quantile of $F_{n}$,
\begin{align*}
  F^{-1}_{n,a}\left( 1- \frac{n\alpha}{|\mathcal{I}_{2,a}|} \right)
  & = F^{-1}_{n}\left( 1- \alpha \right).
\end{align*}
However, in general the quantiles of $F_{n,a}$ will be smaller and
\begin{align*}
  F^{-1}_{n,a}\left( 1- \frac{n\alpha}{|\mathcal{I}_{2,a}|} \right)
  & \leq F^{-1}_{n}\left( 1- \alpha \right).
\end{align*}
This completes the proof of \eqref{eq:quantile-sandwich}. By combining this with the previous result that all empirical quantiles of $E^{\text{CQR}}_i$ and $\tilde{E}^{\text{CQR}}_i$, for $i \in \mathcal{I}_{2,a}$, are within $\eta_n^{1/3}$, we obtain:
\begin{align*}
  \tilde{F}^{-1}_{n,a}\left( 1- \frac{n\alpha}{|\mathcal{I}_{2,a}|} \right) - \eta_n^{1/3}
  & \leq F^{-1}_{n}\left( 1- \alpha \right)
    \leq \tilde{F}^{-1}_{n,a}\left( 1- \frac{n\alpha -|\mathcal{I}_{2,b}| }{|\mathcal{I}_{2,a}|} \right) + \eta_n^{1/3}.
\end{align*}
Recall that we have defined $\hat{Q}_{1-\alpha}(E;\mathcal{I}_2) = F^{-1}_{n}\left( 1- \alpha_n \right)$, where $\alpha_n = \alpha - (1-\alpha)/n$.
Therefore,
\begin{align*}
  \tilde{F}^{-1}_{n,a}\left( 1- \frac{n\alpha_n}{|\mathcal{I}_{2,a}|} \right) - \eta_n^{1/3}
  & \leq \hat{Q}_{1-\alpha}(E;\mathcal{I}_2)
    \leq \tilde{F}^{-1}_{n,a}\left( 1- \frac{n\alpha_n -|\mathcal{I}_{2,b}| }{|\mathcal{I}_{2,a}|} \right) + \eta_n^{1/3}.
\end{align*}
Note that $Q_{1-\alpha}(\tilde{E}^{\text{CQR}}) = 0$. This follows immediately from its definition:
\begin{align*}
  \P{\tilde{E}^{\text{CQR}} \leq 0}
  & = \P{\max\{ q_{\alphalo}(X_{n+1}) - Y, Y - q_{\alphaup}(X_{n+1}) \} \leq 0} \\
  & = \P{ Y \in [q_{\alphalo}(X_{n+1}), q_{\alphaup}(X_{n+1})] }
    = 1-\alpha.
\end{align*}
Then we know from Assumptions~\ref{assumption:iid}--\ref{assumption:regularity} and classical asymptotic theory \cite[Chapter~21]{vaart_1998} that $\tilde{F}^{-1}_{n,a}(1-\alpha)$ is within $o_{\mathbb{P}}(1)$ of the upper $\alpha$ population quantile of $\tilde{E}^{\text{CQR}}_i$, which we denote by $Q_{1-\alpha}(\tilde{E}^{\text{CQR}})$.
Therefore, since both $n\alpha_n/|\mathcal{I}_{2,a}|$ and $(n\alpha_n -|\mathcal{I}_{2,b}|)/|\mathcal{I}_{2,a}|$ are  $\alpha + o_{\mathbb{P}}(1)$, we have
\begin{align*}
  \left| \hat{Q}_{1-\alpha}(E;\mathcal{I}_2) - Q_{1-\alpha}(\tilde{E}^{\text{CQR}}) \right| = o_{\mathbb{P}}(1).
\end{align*}
\end{enumerate}
This completes the proof in the case of CQR. The proof for CQR-m and CQR-r are analogous.

\end{proof}

\clearpage
\section{Supplementary figures} \label{app:plots}

\begin{figure}[!htb]
  \centering
  \includegraphics[]{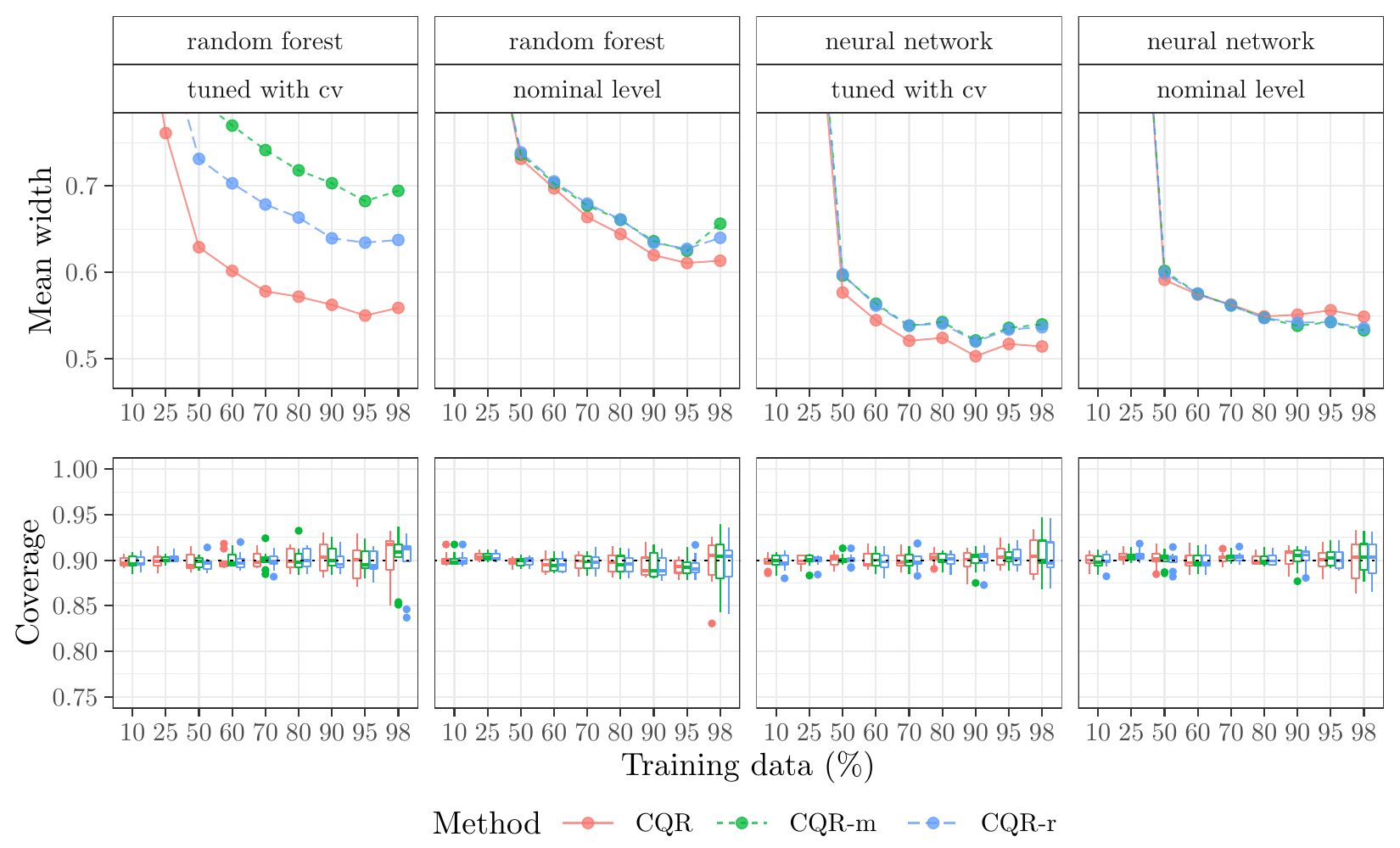}
  \caption{Results on the \textit{bike} data. Other details as in Figure~\ref{fig:community}. The plots are truncated from above on the vertical axis to focus on the most interesting region. In the second through fourth plots on top, the curves corresponding to CQR-m and CQR-r are almost overlapping.}
  \label{fig:bike}
\end{figure}

\begin{figure}[!htb]
  \centering
  \includegraphics[]{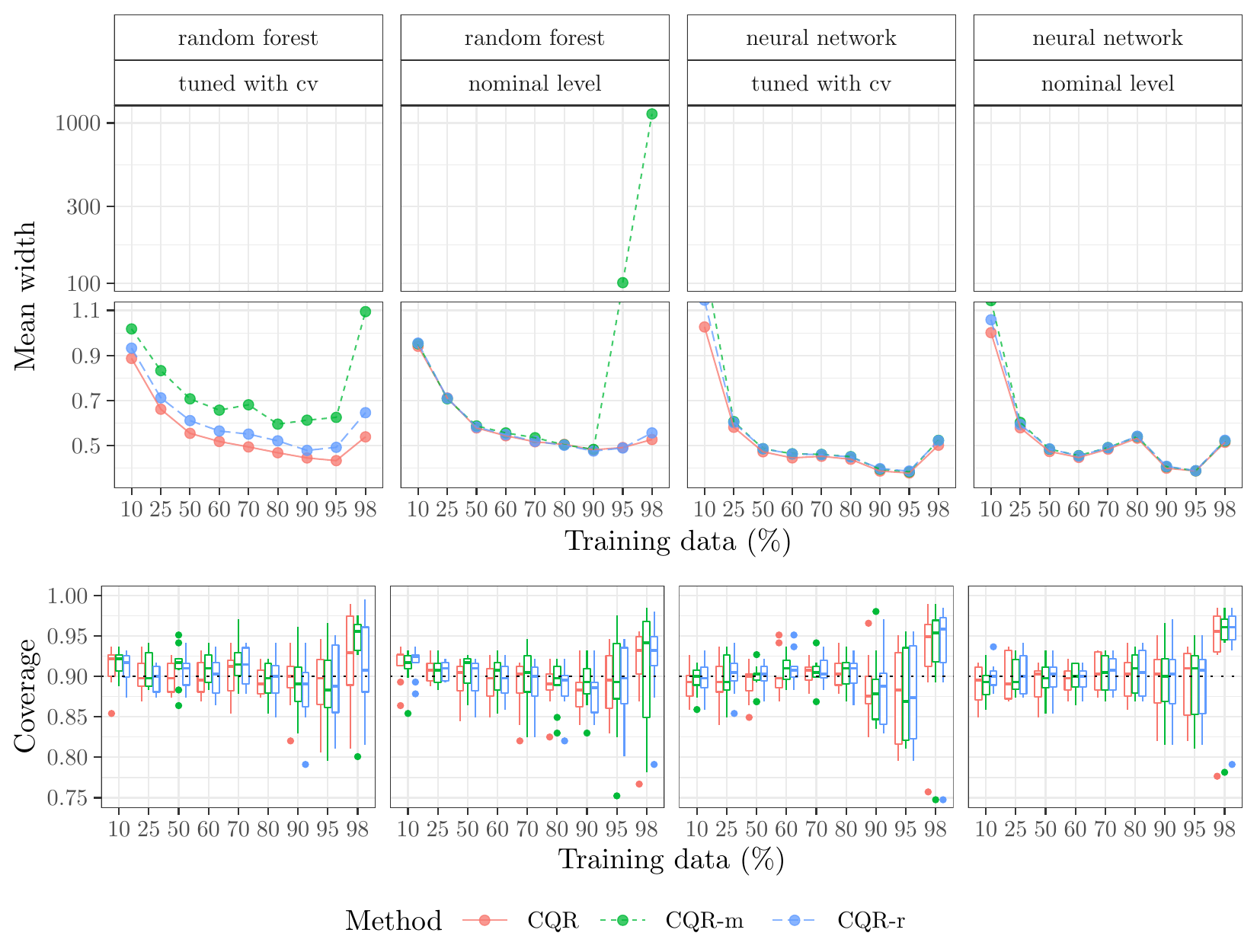}
  \caption{Results on the \textit{concrete} data. Other details as in Figure~\ref{fig:bike}. In the second through fourth plots on top, the three curves are mostly overlapping. The vertical axis in the upper panels is discontinuous to facilitate the visualization of values on different scales.
  }
  \label{fig:concrete}
\end{figure}

\begin{figure}[!htb]
  \centering
  \includegraphics[]{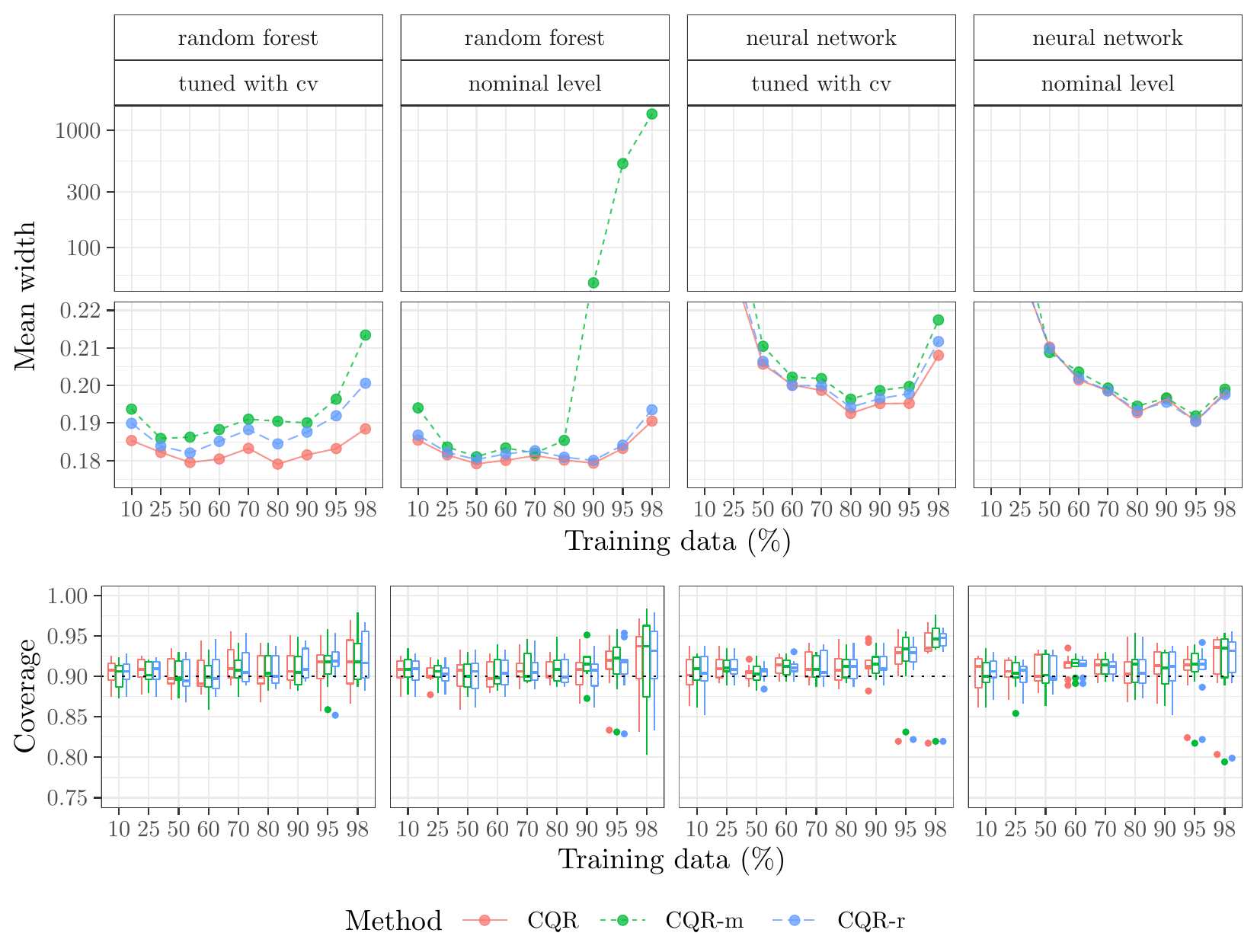}
  \caption{Results on the \textit{star} data. Other details as in Figure~\ref{fig:bike}. In the second through fourth plots on top, the three curves are mostly overlapping. The vertical axis in the upper panels is discontinuous to facilitate the visualization of values on different scales.
  }
  \label{fig:star}
\end{figure}

\begin{figure}[!htb]
  \centering
  \includegraphics[]{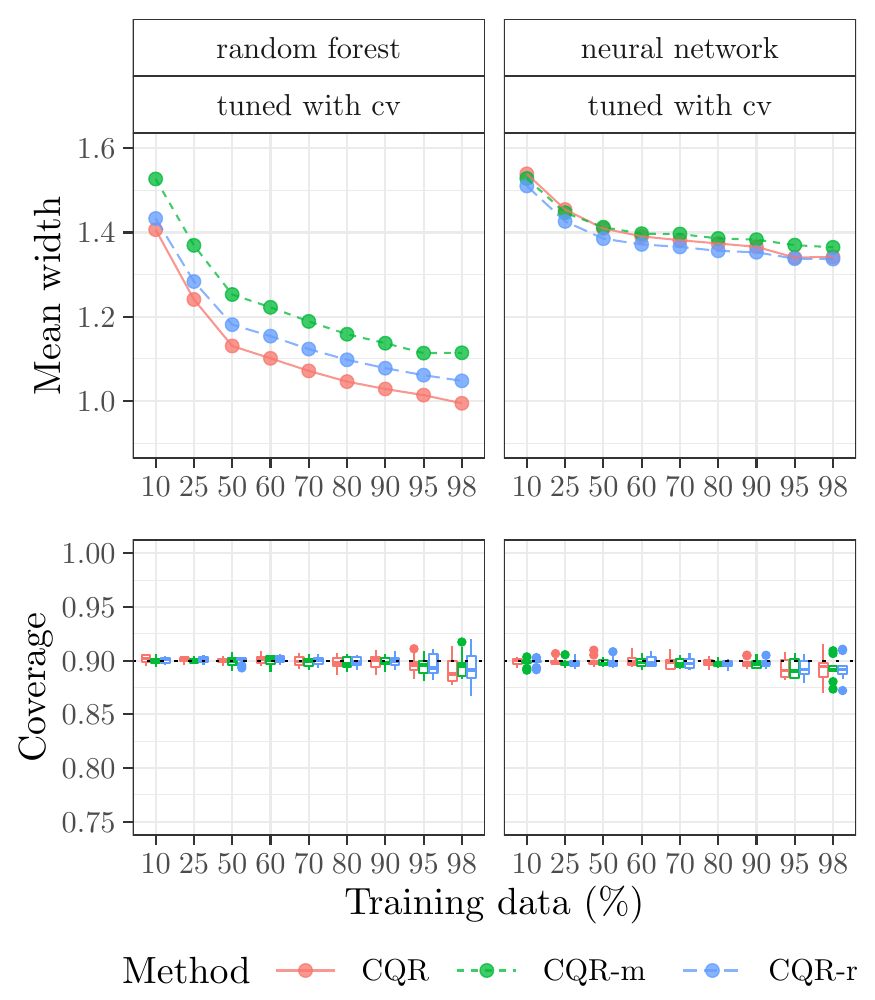}
  \caption{Results on the \textit{bio} data. Other details as in Figure~\ref{fig:bike}.
  }
  \label{fig:bio}
\end{figure}

\begin{figure}[!htb]
  \centering
  \includegraphics[]{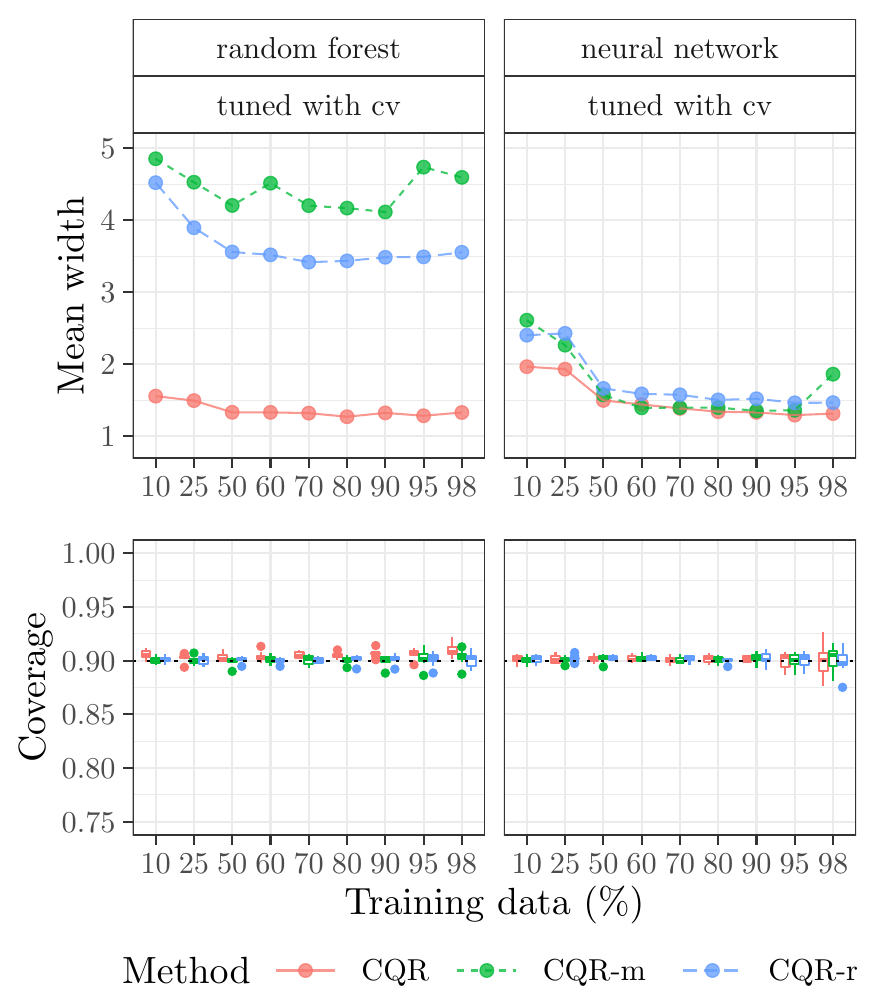}
  \caption{Results on the \textit{blog} data. Other details as in Figure~\ref{fig:bike}.
  }
  \label{fig:blog_data}
\end{figure}

\begin{figure}[!htb]
    \centering
    \begin{subfigure}[t]{0.5\textwidth}
        \centering
        \includegraphics[]{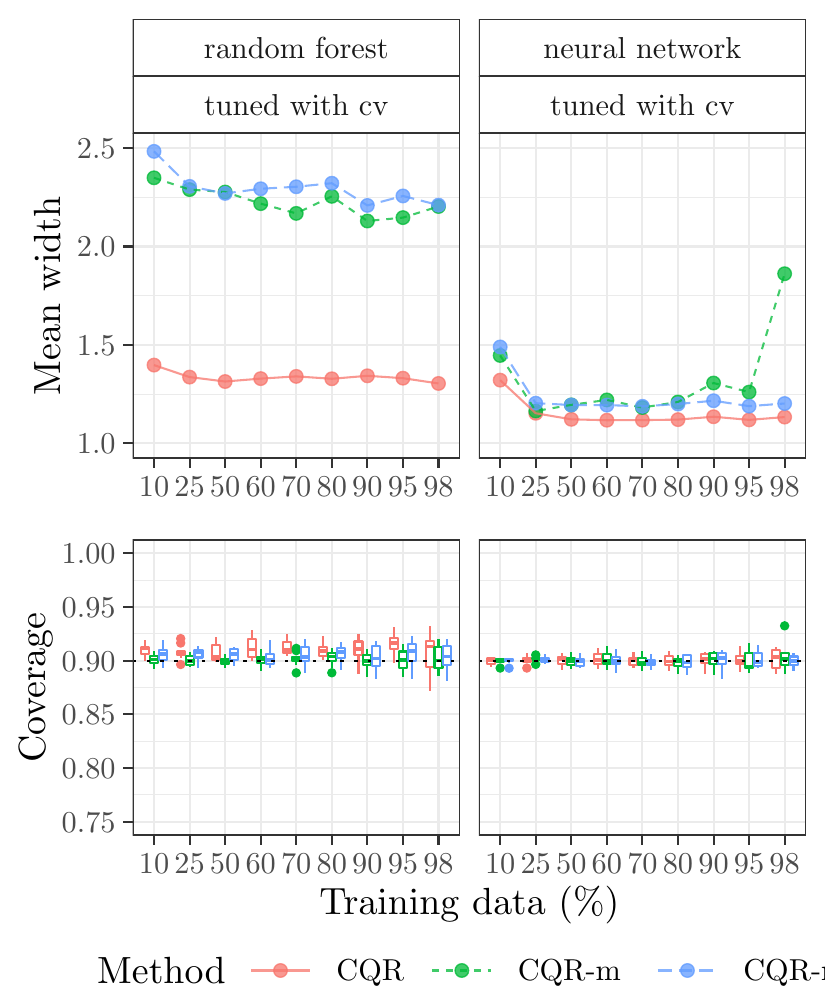}
        \caption{facebook-1}
    \end{subfigure}%
    ~
    \begin{subfigure}[t]{0.5\textwidth}
        \centering
        \includegraphics[]{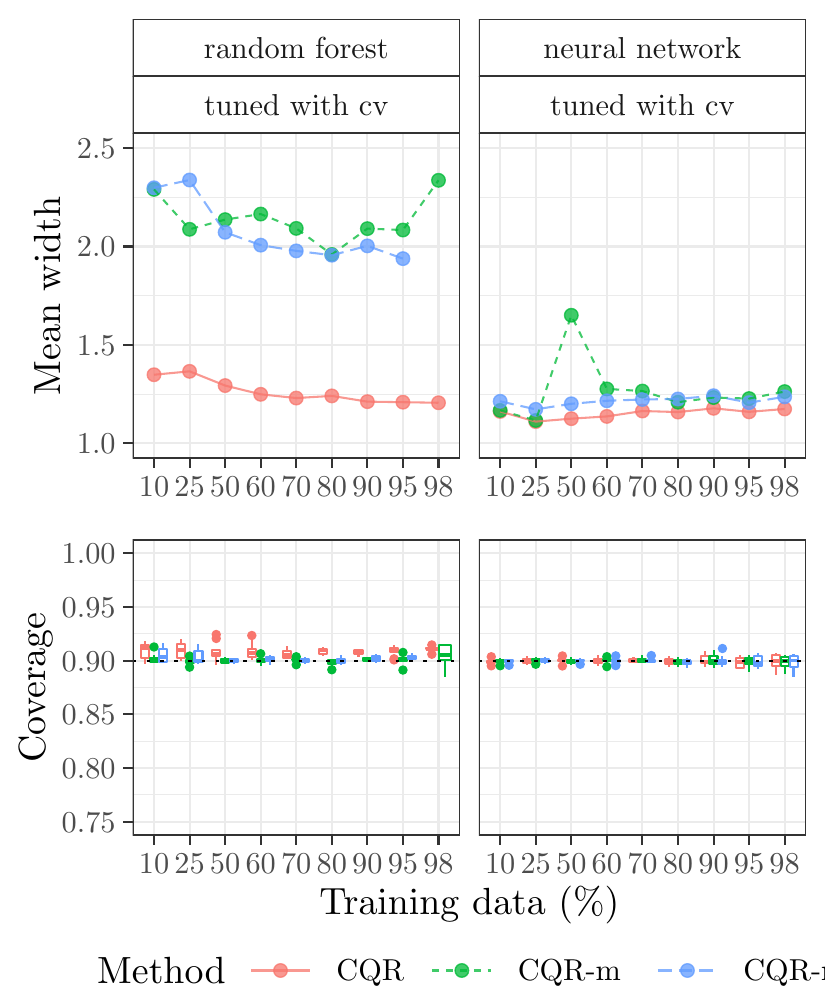}
        \caption{facebook-2}
    \end{subfigure}
  \caption{Results on the \textit{facebook} data. Other details as in Figure~\ref{fig:bike}.
  }
  \label{fig:facebook}
\end{figure}

\begin{figure}[!htb]
    \centering
    \begin{subfigure}[t]{0.5\textwidth}
        \centering
        \includegraphics[]{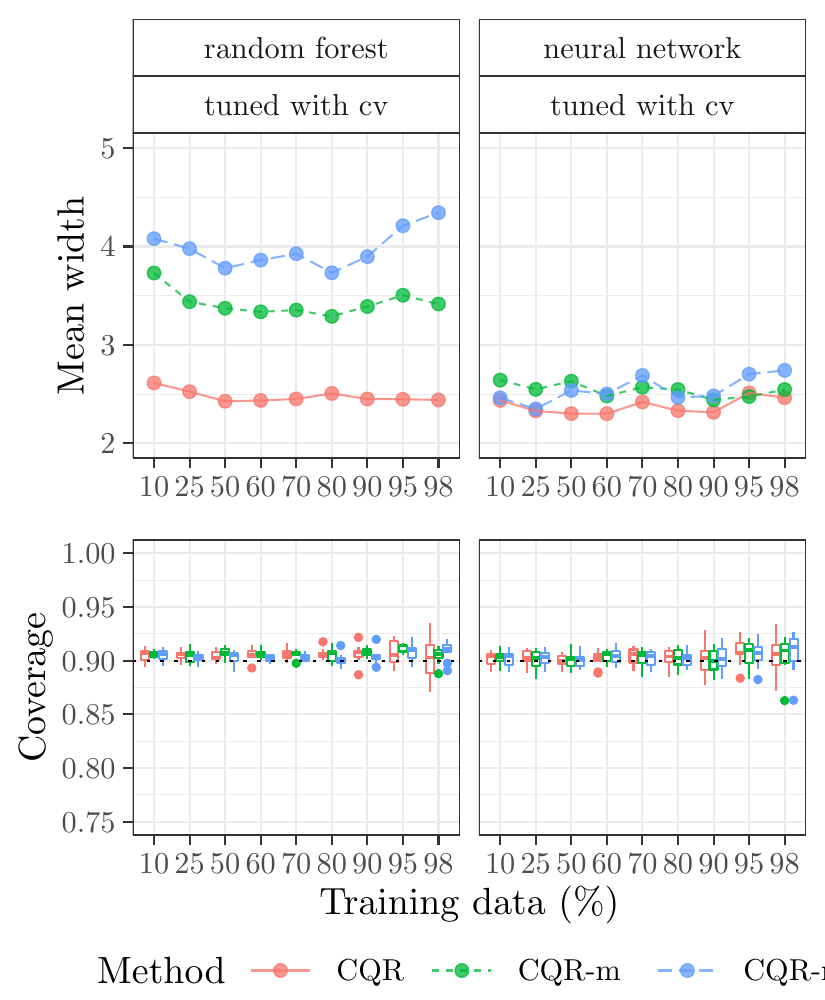}
        \caption{meps-19}
    \end{subfigure}%
    ~
    \begin{subfigure}[t]{0.5\textwidth}
        \centering
        \includegraphics[]{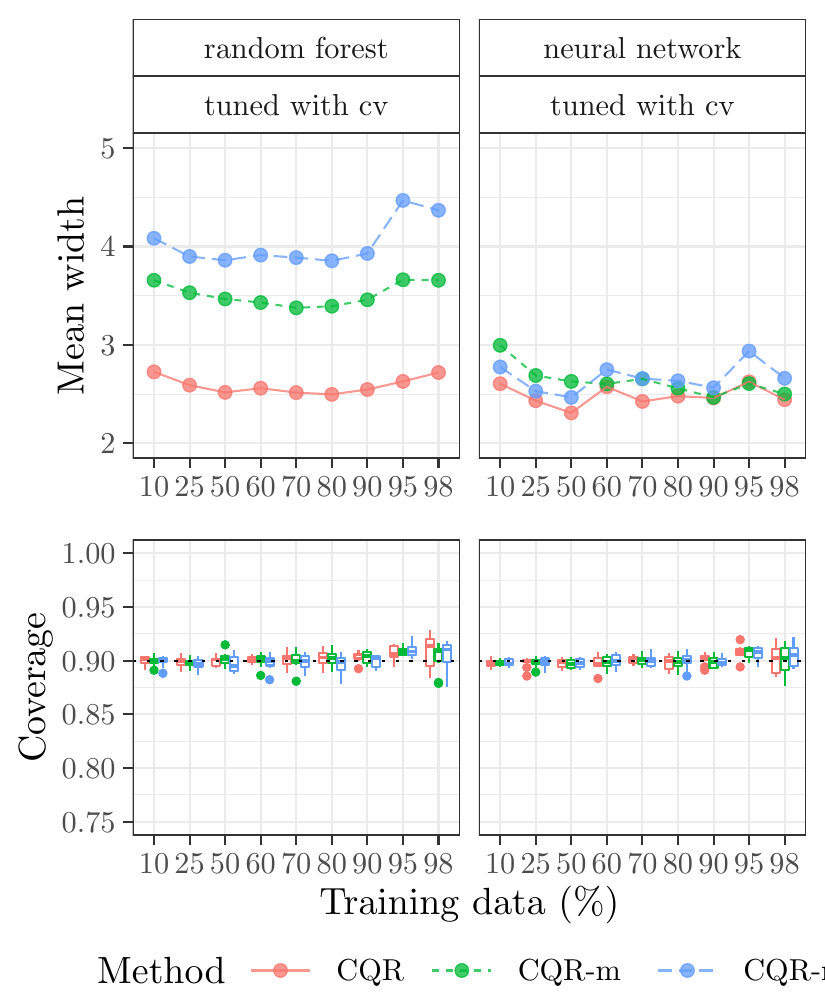}
        \caption{meps-20}
    \end{subfigure}
    ~
    \begin{subfigure}[t]{0.5\textwidth}
        \centering
        \includegraphics[]{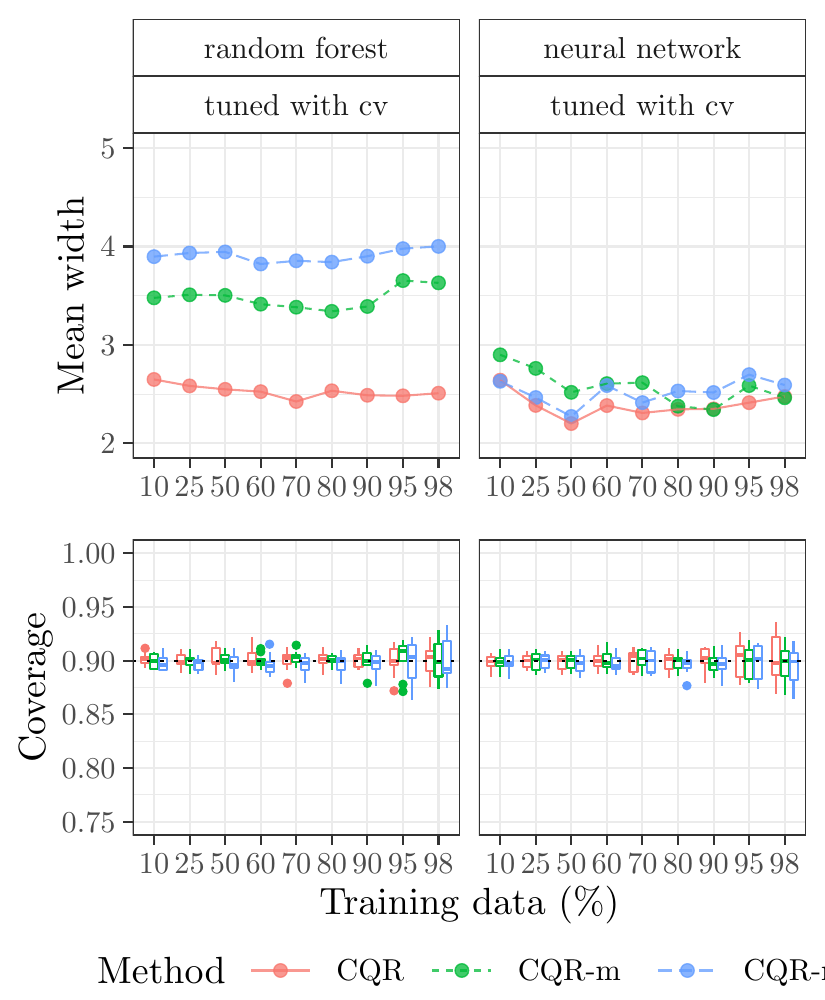}
        \caption{meps-21}
    \end{subfigure}
  \caption{Results on the \textit{meps} data. Other details as in Figure~\ref{fig:bike}.
  }
  \label{fig:meps}
\end{figure}

\end{document}